\documentclass{article}
\usepackage{amsmath,amssymb,amsthm} 

\usepackage{hyperref}

\usepackage{url,xspace}
\usepackage{bussproofs}
\usepackage{mathbbol}
\usepackage[all]{xy}
\SelectTips{cm}{}
\usepackage{stmaryrd}
\usepackage{tikz}

\theoremstyle{plain}
\newtheorem{theorem}{Theorem}[section]
\theoremstyle{definition}
\newtheorem{definition}[theorem]{Definition}
\newtheorem{example}[theorem]{Example}
\newtheorem{remark}[theorem]{Remark}

\newcommand{\pure}{{(0)}} 
\newcommand{\kl}{{(1)}} 
\newcommand{\cst}{{(1)}} 
\newcommand{\acc}{{(1)}} 
\newcommand{\cokl}{{(2)}} 
\newcommand{\modi}{{(2)}} 
\newcommand{\gen}{{(g)}} 
\newcommand{\dec}{{(d)}} 
 
\newcommand{\decp}{{(d')}} 
\newcommand{\decpp}{{(d'')}}  
\newcommand{\done}{{(d_1)}}  
\newcommand{\dtwo}{{(d_2)}}  
 
\newcommand{\eqs}{\cong} 
\newcommand{\eqw}{\sim}
\newcommand{\geqaux}{\gg}
\newcommand{\leqaux}{\ll}
 
\newcommand{\id}{\mathit{id}}
\newcommand{\cC}{\mathbf{C}} 
\newcommand{\cD}{\mathcal{D}} 
\newcommand{\cE}{\mathcal{E}}

\newcommand{\pr}{\mathit{pr}} 
\newcommand{\copr}{\mathit{in}} 
\newcommand{\pair}[1]{\langle #1 \rangle} 
\newcommand{\lpair}[1]{\langle #1 \rangle_l} 
\newcommand{\rpair}[1]{\langle #1 \rangle_r} 
\newcommand{\pa}{\langle \; \rangle} 
\newcommand{\copair}[1]{[ #1 ]} 
\newcommand{\lcopair}[1]{[ #1 ]_l} 
 
\newcommand{\copa}{[\;]} 
\newcommand{\eps}{\varepsilon} 
\newcommand{\To}{\Rightarrow}

\newcommand{\eqn}{\mathit{eq}}
\newcommand{\mon}{\mathit{mon}}
\newcommand{\exc}{\mathit{exc}}
\newcommand{\comon}{\mathit{comon}}
\newcommand{\sta}{\mathit{st}}
 
\newcommand{\empt}{\mathbb{0}}
\newcommand{\unit}{\mathbb{1}}
 
\newcommand{\tagg}{\mathtt{tag}}
\newcommand{\untag}{\mathtt{untag}}
\newcommand{\throw}{\mathtt{throw}}
\newcommand{\try}{\mathtt{try}}
\newcommand{\catch}{\mathtt{catch}}

\newcommand{\lookup}{\mathtt{lookup}}
\newcommand{\update}{\mathtt{update}}
\newcommand{\squad}{\;\;}
\newcommand{\stimes}{\!\times\!}
\newcommand{\splus}{\!+\!}
\newcommand{\scolon}{\!\colon\!}
\newcommand{\scirc}{\!\circ\!}
\newcommand{\sto}{\!\to\!}
\newcommand{\Ename}{\mathit{Exn}}
\newcommand{\Loc}{\mathit{Loc}}
\newcommand{\abr}{\boxed{\mathit{abrupt}}}
\newcommand{\nor}{\boxed{\mathit{normal}}}
\newcommand{\abrupt}{\mathit{abrupt}}
\newcommand{\normal}{\mathit{normal}}
\newcommand{\M}{M} 
\newcommand{\MM}{D}  
\newcommand{\coM}{D} 
\newcommand{\isexc}{\mathit{exc?}} 
 
\newcommand{\Log}{\mathcal{L}} 
\newcommand{\ttr}{\mathit{true}}
\newcommand{\ff}{\mathit{false}}

\newcommand{\rul}[1]{\textrm{(#1)}}

\newcommand{\catchn}[4]{\mathtt{catch}\,(#1\!\To\! #2|\dots|#3\!\To\! #4)}
\newcommand{\all}{\mathtt{all}}

\title{Breaking a monad-comonad symmetry between computational effects}
\date{February 4., 2014}
\author{Jean-Guillaume Dumas \\ 
  Laboratoire Jean Kuntzmann, University of Grenoble, France
        \and Dominique Duval \\ 
  Laboratoire Jean Kuntzmann, University of Grenoble, France
  \and Jean-Claude Reynaud \\
  Reynaud Consulting, Claix, France
}

\begin{document}

\maketitle     

\begin{abstract}
Computational effects may often be interpreted in the Kleisli category 
of a monad or in the coKleisli category of a comonad. 
The duality between monads and comonads corresponds, in general,
to a symmetry between construction and observation, 
for instance between raising an exception and looking up a state. 
Thanks to the properties of adjunction one may go one step further: 
the coKleisli-on-Kleisli category of a monad provides 
a kind of observation with respect to a given construction, 
while dually the Kleisli-on-coKleisli category of a comonad provides 
a kind of construction with respect to a given observation.
In the previous examples this gives rise to catching an exception 
and updating a state. 
However, the interpretation of computational effects 
is usually based on a category which is not self-dual, 
like the category of sets. This leads to a breaking of 
the monad-comonad duality. For instance, in a distributive category 
the state effect has much better properties than the exception effect. 
This remark provides a novel point of view on
the usual mechanism for handling exceptions. 
The aim of this paper is to build an equational semantics 
for handling exceptions based on the coKleisli-on-Kleisli category 
of the monad of exceptions. 
We focus on n-ary functions and conditionals. 
We propose a programmer's language for exceptions 
and we prove that it has the required behaviour 
with respect to n-ary functions and conditionals. 

\end{abstract}

\section{Introduction} 

Categorical semantics for programming languages 
interprets types as objects and terms as morphisms;
in this setting, substitution is composition, 
categorical products are used for dealing with $n$-ary operations
and coproducts for conditionals. 
The most famous result in this direction is the Curry-Howard-Lambek 
correspondence which relates 
intuitionistic logic, simply typed lambda calculus 
and cartesian closed categories.  
The \emph{algebraic effects} challenge is the search for 
some extension of this correspondence 
to a categorical framework corresponding to 
\emph{computational effects}, which means, roughly, 
to non-functional features of programming languages. 
Moggi proposed to use the categorical notion of \emph{monad} 
for this purpose \cite{Moggi91},
then monads were popularized by Wadler \cite{Wa92} 
and implemented in Haskell and F$\sharp$. 
Related categorical notions like 
Freyd categories, arrows, Lawvere theories, were also proposed
\cite{PR97,Hugues00,PP02,HPP06}. 
Moreover, the dual notion of \emph{comonad} can also 
be used for dealing with computational effects 
\cite{UV08,DDFR12-state,POM13}. 
This gives rise to a three-tier classification of terms 
which is similar to the one in \cite{Wa05}. 
Some effects, like the state effect, can be seen 
both as monads and as comonads \cite{Moggi91,DDFR12-state}. 
Other effects, like the handling of exceptions, 
do not fit easily in the monad approach \cite{PP03,PP09,BP13}.
However the use of the co-Eilenberg-Moore-on-Eilenberg-Moore category 
of the monad of exceptions was successfully used by Levy 
for adapting the monad approach to the handling of exceptions \cite{Le06};
in this paper we follow a similar line. 

The aim of this paper is to build an equational semantics 
for handling exceptions based on the coKleisli-on-Kleisli category 
of the monad of exceptions. 
We focus on $n$-ary functions and conditionals 
because they correspond to the dual categorical notions of 
products and coproducts, although their behaviour 
with respect to effects is quite different: in general 
there is no ambiguity in using conditionals with effects, 
whereas the value of an expression involving $n$-ary functions 
may depend on the order of evaluation of the arguments. 
The equational semantics we use are \emph{decorated}: 
the terms and equations are annotated,
in a way similar to the type-and-effect systems \cite{LucassenGifford88}, 
in order to classify them according to their interaction with the effect. 
Typically, for exceptions, terms are classified as pure,
propagators (which must propagate exceptions) 
and catchers (which may recover from exceptions);
thus, there is no need for an explicit ``type of exceptions'', 
and we get a clear dictinction between a coproduct type $A+B$ in the syntax 
and a coproduct $A+E$ (where $E$ is the ``object of exceptions'')
which may be used for interpreting terms involving the type $A$.
In the equational semantics we use coproduct types $A+B$
but we never use coproducts involving $E$.
In order to get an equational semantics for exceptions, 
we start from two facts:
first, the Kleisli-on-coKleisli category of a comonad can be used 
for building an equational semantics for states \cite{DDFR12-state};
secondly, there is a duality between the denotational semantics 
of the state effect and the denotational semantics 
of the core operations for the exception effect \cite{DDFR12-dual}. 
We adapt this duality to the equational level,
then we build the programmer's language for exceptions 
by adding some control to the core operations.
Finally we propose a programmer's language for exceptions, 
built from categorical products and coproducts,  
and we prove that it satisfies equations providing the required behaviour 
(as explained above) with respect to $n$-ary functions and conditionals. 
The equational semantics for states is implemented in Coq \cite{DDEP13-coq}
and the one for exceptions is in progress. 

The duality between monads and comonads corresponds, in general,
to a symmetry between construction and observation:
raising an exception is a construction, 
reading the value of a location is an observation. 
As recalled in Section~\ref{sec:pre},  
thanks to the properties of adjunction one may go one step further: 
the coKleisli-on-Kleisli category of a monad provides 
a kind of observation with respect to a given construction, 
while dually the Kleisli-on-coKleisli category of a comonad provides 
a kind of construction with respect to a given observation.
In the previous examples this gives rise to catching an exception 
and updating a state, respectively.
The  coKleisli-on-Kleisli category of a monad,
as well as the Kleisli-on-coKleisli category of a comonad, 
provide a classification of terms and equations. 
In Section~\ref{sec:dual} we define variants of the equational logic 
for dealing with this classification. 
These variants are called {\em decorated logics}: 
there is a decorated logic $\Log_\mon$ for a monad 
and dually a decorated logic $\Log_\comon$ for a comonad. 
When the monad is the exception monad, 
we can add to the decorated logic $\Log_\mon$ the {\em core} operations 
for exceptions: 
the \emph{tagging} operations 
for encapsulating an ordinary value into an exception, 
and the \emph{untagging} operations 
for recovering the ordinary value which has been previously 
encapsulated in an exception. 
Dually, When the comonad is the state comonad, 
we can add the basic operations for states: 
the \emph{lookup} operations which observe the state 
and the \emph{update} operations which modify it.
In Section~\ref{sec:break} we assume that the category $\cC$ 
has some distributivity or extensivity property,
like for instance the category of sets. 
This breaks the monad-comonad duality: 
the state effect gets better properties with respect to coproducts, while 
the exception effect does not get better properties with respect to products. 
On the comonad side, 
we check that the side-effects due to the evolution of state 
do not perturb the case-distinction features, 
and we provide decorated equations for imposing an order on the 
interpretation of the arguments of multivariate functions.
On the monad side,  
we check that the properties of operations for catching exceptions 
are quite poor. This is circumvented by encapsulating 
the catching operations in {\em try-catch} blocks. 
This provides a novel point of view on the formalization 
of the usual mechanism for handling exceptions. 
We get a programmer's language for exceptions
which has the required behaviour 
with respect to $n$-ary functions and conditionals. 

\section{Preliminaries} 
\label{sec:pre}

We present some well-known results 
about monads and comonads in Section~\ref{ssec:kleisli}) 
and (independently) about equational logic with conditionals 
in Section~\ref{ssec:eqn}. 

\subsection{CoKleisli-on-Kleisli category} 
\label{ssec:kleisli}

This Section relies on \cite{ML}. 
A similar construction is used in \cite{Le06,Jacobs13}, 
with ``Kleisli'' replaced by ``Eilenberg-Moore''. 
Let $\cC$ be a category and $(\M,\eta,\mu)$ a monad on $\cC$. 
Let $\cC^\cst$ be the Kleisli category of this monad  
and $F_0\dashv G_0\colon \cC^\kl \to \cC $ the corresponding adjunction. 
Then $\M=G_0\circ F_0\colon \cC \to \cC$.
Let $\MM=F_0\circ G_0\colon \cC^\kl \to \cC^\kl$,
it is the endofunctor of a comonad $(\MM,\eps,\delta)$ on $\cC^\kl$.
Let $\cC^\cokl$ be the coKleisli category of this comonad 
and $F_1\dashv G_1\colon \cC^\kl \to \cC^\cokl $ the corresponding adjunction. 
Then $\MM=F_1\circ G_1\colon \cC^\kl \to \cC^\kl$. 
In such a situation, there is a unique functor $K\colon \cC^\cokl \to \cC$ 
such that $K\circ G_1=G_0$ and $F_0\circ K=F_1$. 
$$
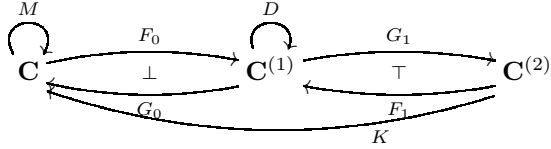
\xymatrix@C=6pc{
\cC \ar@(ul,ur)^{\M} \ar@<+.5ex>@/^/[r]^{F_0} &
\cC^\kl \ar@(ul,ur)^{\MM} \ar@<+.5ex>@/^/[l]^{G_0} \ar@ {}[l]|{\bot} 
  \ar@<+.5ex>@/^/[r]^{G_1} \ar@ {}[r]|{\top} &
\cC^\cokl \ar@<+.5ex>@/^/[l]^{F_1} \ar@<+1ex>@/^4ex/[ll]^(.3){K} \\
}$$
The three categories $\cC$, $\cC^\kl$ and $\cC^\cokl$ have the same objects.
There is a morphism $g^\kl\colon A\to B$ in $\cC^\kl$ 
for each morphism $g_1\colon A\to \M B$ in $\cC$, 
and there is a morphism $h^\cokl\colon A\to B$ in $\cC^\cokl$ 
for each morphism $h_2\colon \M A\to \M B$ in $\cC$.
The  functor $K$ maps $A$ to $\M A$ and 
$h^\cokl\colon A\to B$ to $h_2\colon \M A\to \M B$. 
We are mainly interested in the functors $F_0$ and $G_1$.
They are the identity on objects, 
$F_0$ maps $f_0\colon A\to B$ in $\cC$
to $f^\kl\colon A\to B$ in $\cC^\kl$ corresponding to 
$f_1=\eta_B\circ f_0\colon A\to \M B$ in $\cC$, and 
$G_1$ maps $g^\kl\colon A\to B$ in $\cC^\kl$ corresponding to 
$g_1\colon A\to \M B$ in $\cC$
to $g^\cokl\colon A\to B$ in $\cC^\cokl$ corresponding to 
$g_2=\mu_B\circ \M g_1\colon \M A\to \M B$ in $\cC$. 
Thus, $G_1\circ F_0$ maps $f_0\colon A\to B$ in $\cC$
to $f^\cokl\colon A\to B$ in $\cC^\cokl$ corresponding to 
$f_2= \M f_0\colon \M A\to \M B$ in $\cC$. 

\subsection{Equational logic with conditionals}
\label{ssec:eqn}

We choose a categorical presentation of logic as for instance in \cite{Pitts},
in a {\em bicartesian} category 
(i.e., a category with finite products and coproducts). 
In a functional programming language, from this point of view, 
types are objects, terms are morphisms and substitution is composition.
Each term $f$ has a source type $A$ and a target type $B$,
this is denoted $f\colon A\to B$. 
A term has precisely one source type,  
which can be a product type or the unit type $\unit$. 
A $n$-ary operation $f:A_1,\dots,A_n\to B$ 
corresponds to a morphism $f:A_1\times\dots\times A_n\to B$ 
(this holds for every $n\geq0$, with $f:\unit\to B$ when $n=0$).
Typically, when $n=2$, the substitution of terms 
$a_1\colon A\to A_1,a_2\colon A\to A_2$ for the variables $x_1,x_2$ 
in $f(x_1,x_2)$ is the composition of the pair 
$\pair{a_1,a_2}\colon A\to A_1\times A_2$ with $f:A_1\times A_2\to B$. 
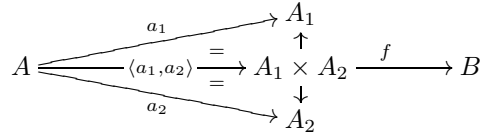
$$ \xymatrix@R=.5pc@C=3pc{
&& A_1 &\\
A \ar[rru]^{a_1} \ar[rrd]_{a_2} \ar[rr]|{\,\pair{a_1,a_2}\,} && 
  A_1\times A_2 \ar[u] \ar[d] \ar@{}[ull]|(.3){=} \ar@{}[dll]|(.3){=} 
    \ar[r]^f & B \\
&& A_2 &\\
}$$ 
Conditionals corresponds to copairs: 
a command like {\it if $b$ then $f$ else $g$} corresponds to 
the morphism $\copair{f|g}\circ b$, 
where $\copair{f|g}$ is the \emph{copair} of $f\colon \unit \to B$ 
and $g\colon \unit \to B$, i.e., 
the unique morphism $h\colon \unit+\unit \to B$ such that 
$h\circ \ttr = f$ and $h\circ \ff = g$. 
$$ \xymatrix@R=.5pc@C=3pc{
& \unit \ar[d]_{\ttr} \ar@<.5ex>[rrd]^{f} && \\
A \ar[r]^{b} & \unit+\unit \ar[rr]|{\,\copair{f|g}\,}
    \ar@{}[urr]|(.3){=} \ar@{}[drr]|(.3){=} && B \\ 
& \unit \ar[u]^{\ff} \ar@<-.5ex>[rru]_{g} && \\
}$$ 
The grammar and the rules of the equational logic with conditionals 
are recalled in Fig.~\ref{fig:log-eqn}. 
For short, rules with the same premisses may be grouped together: 
$\frac{H_1\dots H_n}{C_1},...,\frac{H_1\dots H_n}{C_p}$ 
may be written $\frac{H_1\dots H_n}{C_1\dots C_p}$. 
\begin{figure}[!ht]   
\renewcommand{\arraystretch}{1.7}
$$ \begin{array}{|l|} 
\hline
\mbox{Grammar} \\ 
\quad \textrm{Types: } t::= 
 A\mid B\mid \dots\mid 
 t\times t\mid \unit\mid
 t+t\mid\empt  \\
\quad \textrm{Terms: } f::=   \id_t\mid f\circ f\mid  
 \pair{f,f} \mid \pr_{t,t,1}\mid \pr_{t,t,2}\mid \pa_t \mid
 \copair{f|f}\mid \copr_{t,t,1}\mid\copr_{t,t,2}\mid\copa_t \quad \\
\quad \textrm{Equations: } e::=  f\equiv f   \\  
\hline
\mbox{Equivalence rules} \\ 
\quad \rul{refl} \quad 
  \dfrac{f}{f \equiv f} \qquad
\rul{sym} \quad 
  \dfrac{f \equiv g}{g \equiv f}  \qquad
\rul{trans} \quad 
  \dfrac{f \equiv g \squad g \equiv h}{f \equiv h}  \\ 
\hline
\mbox{Categorical rules} \\ 
\quad \rul{id} \quad 
  \dfrac{A}{\id_A\colon A\to A } \qquad 
\rul{comp} \quad 
  \dfrac{f\colon A\to B \quad g\colon B\to C}
    {(g\circ f) \colon A\to C}  \\
\quad \rul{id-source} \quad 
  \dfrac{f\colon A\to B}{f\circ \id_A \equiv f} \qquad 
\rul{id-target} \quad 
  \dfrac{f\colon A\to B}{\id_B\circ f \equiv f} \\
\quad \rul{assoc} \quad 
  \dfrac{f\colon A\to B \squad g\colon B\to C \squad h\colon C\to D}
  {h\circ (g\circ f) \equiv (h\circ g)\circ f}  \\
\hline
\mbox{Congruence rules} \\ 
\quad \rul{repl} \quad 
  \dfrac{f_1\equiv f_2\colon A\to B \squad g\colon B\to C}
    {g\circ f_1 \equiv g\circ f_2 }  \qquad
\rul{subs} \quad 
  \dfrac{f\colon A\to B \squad g_1\equiv g_2\colon B\to C}
    {g_1 \circ f \equiv g_2\circ f } \\
\hline
\mbox{Product rules} \\ 
\quad \rul{prod} \quad 
  \dfrac{B_1 \quad B_2 }
    {\pr_1\colon B_1\stimes B_2 \to B_1 \quad 
    \pr_2\colon B_1\stimes B_2 \to B_2}  \\ 
\quad \rul{pair}  \quad 
  \dfrac{f_1\colon  A \to B_1 \quad f_2\colon  A \to B_2}
    {\pair{f_1,f_2}\colon  A\to B_1\stimes B_2 \quad
    \pr_1\circ\pair{f_1,f_2} \equiv f_1 \quad 
    \pr_2\circ\pair{f_1,f_2} \equiv f_2 }  \\
\quad \rul{pair-u} \quad 
  \dfrac{f_1\scolon A\!\sto\! B_1 \squad
    f_2\scolon A\!\sto\! B_2 \squad 
    g\scolon A\!\sto\! B_1\stimes B_2 \squad 
    \pr_1\circ g\equiv f_1 \squad 
    \pr_2\circ g\equiv f_2 }
    {g \equiv \pair{f_1,f_2}} \\
\quad \rul{final} \quad 
  \dfrac{A}{\pa_A\colon A\to \unit} \qquad
\rul{final-u} \quad 
  \dfrac{f\colon A\to \unit}{f \equiv \pa_A} \\
\hline
\mbox{Coproduct rules} \\ 
\quad \rul{coprod} \quad 
  \dfrac{A_1 \quad A_2}
    {\copr_1\colon A_1\to A_1\splus A_2 \quad 
    \copr_2\colon A_2\to A_1\splus A_2 } \\ 
\quad \rul{copair} \quad 
  \dfrac{f_1\colon  A_1 \to B \quad f_2\colon  A_2 \to B}
    {\copair{f_1|f_2}\colon  A_1\splus A_2 \to B \quad
    \copair{f_1|f_2} \circ \copr_1 \equiv  f_1 \quad 
    \copair{f_1|f_2} \circ \copr_2 \equiv  f_2  } \\
\quad \rul{copair-u} \quad 
  \dfrac{f_1\scolon A_1 \sto B \squad 
    f_2\scolon A_2 \sto B \squad 
    g\scolon A_1\!\splus\! A_2 \sto B \squad 
    g\circ \copr_1 \equiv f_1 \squad 
    g\circ \copr_2 \equiv f_2 }
    {g \equiv \copair{f_1|f_2}}  \\
\quad \rul{initial} \quad 
  \dfrac{B}{\copa_B\colon \empt\to B} \qquad
\rul{initial-u} \quad 
  \dfrac{f\colon \empt\to B}{f \equiv \copa_B} \\
\hline 
\end{array}$$
\renewcommand{\arraystretch}{1}
\caption{$\Log_\eqn$: the equational logic with conditionals} 
\label{fig:log-eqn} 
\end{figure}

\section{The duality} 
\label{sec:dual}

In Sections~\ref{ssec:monad} and~\ref{ssec:comonad} we define 
decorated logics $\Log_\mon$ and $\Log_\comon$, 
together with their interpretation in 
a category with a monad and with a comonad, respectively. 
Then in Sections~\ref{ssec:excore} and~\ref{ssec:state}  
we extend $\Log_\mon$ and $\Log_\comon$ into 
$\Log_\exc$ and $\Log_\sta$ which are dedicated 
to the monad of exception and to the comonad of states, respectively. 
The interpretations of these logics provide the duality 
between the denotational semantics of states and exceptions 
mentioned in \cite{DDFR12-dual}.
All these logics are called \emph{decorated logics}  
because their grammar and inference rules are essentially the 
grammar and inference rules for 
the logic $\Log_\eqn$ (from Section~\ref{ssec:eqn})   
together with \emph{decorations} for the terms and for the equations.
The decorations for the terms are similar to the \emph{annotations} 
of the types and effects systems \cite{LucassenGifford88}. 
Decorated logics are introduced in \cite{DD10-dialog} 
in an abstract categorical framework which will not be explicitly used 
in this paper. 

\subsection{A decorated logic for a monad} 
\label{ssec:monad}

In the logic $\Log_\mon$ for monads, 
each term has a decoration which is denoted as a superscript
$\pure$, $\cst$ or $\modi$:
a term is \emph{pure} when its decoration is $\pure$,
it is a \emph{constructor} when its decoration is $\cst$
and a \emph{modifier} when its decoration is $\modi$. 
Each equation has a decoration which is denoted by 
replacing the symbol $\equiv$ either by $\eqs$ or by $\eqw$:
an equation with $\eqs$ is called \emph{strong},
with $\eqw$ it is called \emph{weak}. 
In order to give a meaning to the logic $\Log_\mon$, 
let us consider a bicartesian category $\cC$ with a monad $(\M,\eta,\mu)$.
The categories $\cC^\pure=\cC$, $\cC^\cst$, $\cC^\modi$ and the 
functors $F_0:\cC^\pure \to \cC^\cst$ and $G_1: \cC^\cst\to \cC^\modi$ 
are defined as in Section~\ref{ssec:kleisli}.
Then we get an interpretation $\cC_\M$ of 
the grammar and the conversion rules of $\Log_\mon$ as follows. 
\begin{itemize}
\item A type $A$ is interpreted as an object $A$ of $\cC$.
\item A term $f^\dec\colon A\!\to\! B$ is interpreted as a
morphism $f\colon A\!\to\! B$ in $\cC^\modi$;
if $d=0$ then $f$ must be in the image of $\cC^\pure$ by $G_1\circ F_0$,
and if $d=1$ then $f$ must be in the image of $\cC^\cst$ by $G_1$. 
This means that all terms are interpreted as morphisms of $\cC$: 
a pure term $f^\pure\colon A\!\to\! B$ as a morphism 
$f_0\colon A\!\to\! B$ in~$\cC$;
a constructor $g^\cst\colon\! A\!\to\! B$ as a morphism 
$g_1\colon\! A\!\to\! \M B$ in~$\cC$;
and a modifier $h^\modi\colon A\!\to\! B$ as a morphism 
$h_2\colon \M A\!\to\! \M B$ in~$\cC$.
\item A strong equation $f^\dec\eqs g^\dec\colon A\to B$ 
is interpreted as an equality 
$f= g\colon A\to B$ in~$\cC^\modi$, i.e., 
as an equality $f_2=g_2\colon \M A\to \M B$ in~$\cC$. 
\item A weak equation $f^\dec\eqw g^\dec\colon A\to B$ 
is interpreted as an equality 
$f_2\circ\eta_A=g_2\circ\eta_A\colon A\to \M B$ in $\cC$.
\end{itemize}
\begin{example}
\label{example:monad-list}
Let us consider the monad of lists (or words),
and its interpretation in the category of sets. 
Then a term $f:A\to B$ is interpreted as a code,
i.e., as a map $f:A^*\to B^*$
from the words on $A$ to the words on $B$. 
The classification of the terms provided by the decorations 
corresponds to a well-known classification of the codes: 
if $f$ is constructor then for each word $u=x_1\dots x_n$ on $A$
the word $f(u)=f(x_1)\dots f(x_n)$ is the concatenation of 
the images of the letters in $u$,
and if $f$ is pure then in addition for each letter $x$ in $A$
the word $f(x)$ is a letter in $B$. 
\end{example}
The inference rules of $\Log_\mon$ 
are decorated versions of the rules of the equational logic with conditionals. 
The main rules are given in Fig.~\ref{fig:log-mon},
and all rules in Appendix~\ref{app:logic}. 
When a decoration is clear from the context, it is often omitted.  
\begin{itemize}
\item 
The conversion rules are decorated versions of rules of the form~$\frac{H}{H}$. 
\item 
All rules of $\Log_\eqn$ are decorated with $\pure$ for terms and $\eqs$ for 
equations: the pure terms with the strong equations form 
a sublogic of $\Log_\mon$, which is the same as $\Log_\eqn$. 
Thus, the structural operations like $\id$, $\pr$, $\pa$, $\copr$, $\copa$,
are pure.
\item 
The congruence rules for equations are decorated with all decorations for terms 
and for equations, with one notable exception:
the substitution rule holds only when the substituted term is pure.
\item 
The categorical rules hold for all decorations and
the decoration of a composed term is 
the maximum of the decorations of its components. 
\item 
The product rules are decorated only as pure. 
\item 
For the coproduct rules, 
the terms in rules $\rul{copair}$ and $\rul{copair-u}$ 
can be decorated as pure or constructors, 
and the decoration of a copair is 
the maximum of the decorations of its components. 
Thus, conditionals can be built from constructors, 
but not from modifiers. 
The decorated rule $\rul{initial-u}$ states that~$\copa_B$ is the unique 
term from $\empt$ to $B$, up to weak equality.
\end{itemize}
\begin{figure}[!ht]   
\renewcommand{\arraystretch}{1.7}
$$ \begin{array}{|l|} 
\hline
\mbox{Conversion rules} \\ 
\quad \rul{pure-acc} \quad \dfrac{f^\pure}{f^\acc} \quad 
  \rul{acc-mod} \quad \dfrac{f^\acc}{f^\modi}  \\
\quad \rul{strong-weak} \quad \dfrac{f^\dec\eqs g^\decp}{f\eqw g} \quad
  \rul{weak-strong} \quad \dfrac{f^\dec\eqw g^\decp}{f\eqs g} 
  (d,d'\leq 1) \quad \\
\hline
\mbox{Weak substitution rule} \\ 
\quad \rul{w-subs} \quad 
  \dfrac{f^\pure\colon A\to B \squad g_1^\dec\eqw g_2^\decp\colon B\to C}
    {g_1 \circ f \eqw g_2\circ f } \\
\hline
\mbox{Coproduct rules} \\ 
\quad \rul{coprod} \quad 
  \dfrac{A_1 \quad A_2}
    {\copr_1^\pure\colon A_1\to A_1\splus A_2 \quad 
    \copr_2^\pure\colon A_2\to A_1\splus A_2 } \\ 
\quad \rul{copair} \; 
  \dfrac{ f_1^\done\colon  A_1 \to B \quad 
    f_2^\dtwo\colon  A_2 \to B}
    {\copair{f_1|f_2}^{(max(d_1,d_2))} \colon A_1\splus A_2 \to B \quad
    \copair{f_1|f_2} \circ \copr_1 \eqs  f_1 \quad 
    \copair{f_1|f_2} \circ \copr_2 \eqs  f_2  } (d_1,d_2\leq1) \\
\quad \rul{copair-u} \; 
  \dfrac{f_1^\done\scolon A_1 \sto B \squad 
    f_2^\dtwo\scolon A_2 \sto B \squad 
    g^\dec\scolon A_1\!\splus\! A_2 \sto B \squad 
    g\scirc \copr_1 \!\eqs\! f_1 \squad 
    g\scirc \copr_2 \!\eqs\! f_2 }
    {g \!\eqs\! \copair{f_1|f_2}} (d_1,d_2,d\!\leq\!1) \\  
\quad \rul{initial} \quad 
  \dfrac{B}{\copa_B^\pure\colon \empt\to B} \qquad
  \rul{initial-u} \quad 
  \dfrac{f^\modi\colon \empt\to B}{f \eqw \copa_B} \\ 
\hline
\end{array}$$
\renewcommand{\arraystretch}{1}
\caption{$\Log_\mon$: some decorated rules for a monad} 
\label{fig:log-mon} 
\end{figure}
It is easy to check that these rules are satisfied by the interpretation 
$\cC_\M$ of $\Log_\mon$. 
Each $f^\pure$ may be converted to $f^\cst=F_0f^\pure$ 
and to $f^\modi=G_1F_0f^\pure$, and each $g^\cst$ to $g^\modi=G_1g^\cst$. 
Each strong equality $f= g$ 
gives rise to an equality $f_2\circ\eta_A=g_2\circ\eta_A$,
and both equalities are equivalent when $f$ and $g$ are in $\cC^\cst$.
Products and coproducts in $\Log_\mon$ 
are interpreted as products and coproducts in $\cC$.  
For instance, the pair of two constructors 
$f^\cst\colon A\to B_1$ and $g^\cst\colon A\to B_2$ 
is interpreted as the pair 
$\pair{f_1,g_1}\colon \M A\to B_1\times B_2$ in~$\cC$. 

\subsection{A decorated logic for a comonad}
\label{ssec:comonad}

The dual of the decorated logic $\Log_\mon$ for a monad 
is the decorated logic $\Log_\comon$ for a comonad. 
Thus, the grammar of $\Log_\comon$ is the same as the grammar of $\Log_\mon$,
but a term with decoration $\acc$ is now called an \emph{accessor} 
(or an \emph{observer}). 
The conversion rules are the same as those in $\Log_\mon$.
Let $\cC$ be a bicartesian category with a comonad $(\coM,\eps,\delta)$.
The categories $\cC^\pure=\cC$, $\cC^\acc$, $\cC^\modi$ and the 
functors $F_0:\cC^\pure \to \cC^\acc$ and $G_1: \cC^\acc\to \cC^\modi$ 
are defined dually to Section~\ref{ssec:kleisli}.
Then we get an interpretation $\cC_\coM$ of 
the grammar of $\Log_\comon$ as follows.
\begin{itemize}
\item A type $A$ is interpreted as an object $A$ of $\cC$.
\item A term $f^\dec\colon A\!\to\! B$ is interpreted as a
morphism $f\colon A\!\to\! B$ in $\cC^\modi$, 
which can be expressed as a morphism in $\cC$: 
a pure term $f^\pure\colon A\!\to\! B$ as a morphism 
$f_0\colon A\!\to\! B$ in~$\cC$;
an accessor $g^\acc\colon\! A\!\to\! B$ as a morphism 
$g_1\colon\! \coM A\!\to\! B$ in~$\cC$;
and a modifier $h^\modi\colon A\!\to\! B$ as a morphism 
$h_2\colon \coM A\!\to\! \coM B$ in~$\cC$.
\item A strong equation $f^\dec\eqs g^\dec\colon A\to B$ 
is interpreted as an equality 
$f_2=g_2\colon \coM A\to \coM B$ in~$\cC$. 
\item A weak equation $f^\dec\eqw g^\dec\colon A\to B$ 
is interpreted as an equality 
$\eps_B\circ f_2=\eps_B\circ g_2\colon A\to \coM B$ in~$\cC$.
\end{itemize}
The rules for $\Log_\comon$ are nearly the same 
as the corresponding rules for $\Log_\mon$, except that
for weak equations the substitution rule always holds 
while the replacement rule holds only when the replaced term is pure,
and in the rules for products and coproducts the decorations are permuted, 
see Fig.~\ref{fig:log-comon} for the  main rules.
\begin{figure}[!ht]   
\renewcommand{\arraystretch}{1.7}
$$ \begin{array}{|l|} 
\hline
\mbox{Conversion rules} \\ 
\quad \dfrac{f^\pure}{f^\acc} \qquad \dfrac{f^\acc}{f^\modi} 
  \qquad \dfrac{f^\dec\eqs g^\decp}{f\eqw g} \mbox{ for all } d,d' \qquad
\dfrac{f^\dec\eqw g^\decp}{f\eqs g} \mbox{ for all } d,d'\leq 1 \\
\hline
\mbox{Weak replacement rule} \\ 
\quad \rul{w-repl} \quad 
  \dfrac{f_1^\dec\eqw f_2^\decp\colon A\to B \squad g^\pure\colon B\to C}
    {g\circ f_1 \eqw g\circ f_2 } \\
\hline
\mbox{Product rules} \\ 
\quad \rul{prod} \quad 
  \dfrac{B_1 \quad B_2 }
    {\pr_1^\pure\colon B_1\stimes B_2 \to B_1 \quad 
    \pr_2^\pure\colon B_1\stimes B_2 \to B_2}  \\ 
\quad \rul{pair}  \; 
  \dfrac{ f_1^\done\colon  A \to B_1 \quad f_2^\dtwo\colon  A \to B_2}
    {\pair{f_1,f_2}^{(max(d_1,d_2))}\colon  A\to B_1\stimes B_2 \quad 
    \pr_1\circ\pair{f_1,f_2} \eqs f_1 \quad 
    \pr_2\circ\pair{f_1,f_2} \eqs f_2 } (d_1,d_2\leq1) \\
\quad \rul{pair-u} \;
  \dfrac{f_1^\done\scolon A\!\sto\! B_1 \squad
    f_2^\dtwo\scolon A\!\sto\! B_2 \squad 
    g^\dec\scolon A\!\sto\! B_1\stimes B_2 \squad 
    \pr_1\circ g\eqs f_1 \squad 
    \pr_2\circ g\eqs f_2 }
    {g \eqs \pair{f_1,f_2}} (d_1,d_2,d\!\leq\!1) \\
\quad \rul{final} \quad 
  \dfrac{A}{\pa_A^\pure\colon A\to \unit} \qquad
   \rul{final-u} \quad 
  \dfrac{f^\modi\colon A\to \unit}{f \eqw \pa_A} \\
\hline
\end{array}$$
\renewcommand{\arraystretch}{1}
\caption{$\Log_\comon$: some decorated rules for a comonad} 
\label{fig:log-comon} 
\end{figure}
The logic $\Log_\comon$ can be interpreted dually to $\Log_\mon$.
Let $\cC$ be a bicartesian category 
and $(\coM,\eps,\delta)$ a comonad on $\cC$.  
Then we get a model $\cC_{\coM}$ of the decorated logic $\Log_\comon$, where 
an accessor $f^\acc\colon\! A\!\to\! B$ 
is interpreted as a morphism $f_1\colon \coM A\to B$ 
in~$\cC$, 
a weak equation $f^\modi\eqw g^\modi\colon A\to B$ as an equality 
$\eps_B\circ f_2=\eps_B\circ g_2\colon \coM A\to B$ 
in $\cC$ 
and a copair of two accessors 
$f^\acc\colon A_1\to B$ and $g^\acc\colon A_2\to B$ 
as the copair 
$\copair{f_1|g_1}\colon A_1+A_2\to \coM B$ in~$\cC$. 

\subsection{A decorated logic for the monad of exceptions} 
\label{ssec:excore}

Let us assume that there is in $\cC$ 
a distinguished object $E$ called the \emph{object of exceptions}.
The \emph{monad of exceptions} on $\cC$ is the monad $(\M,\eta,\mu)$  
with endofunctor $\M X=X+E$, its unit $\eta$ is made of the coprojections 
$\eta_X\colon X\to X+E$ and its multiplication $\mu$ 
is defined by 
$\mu_X=\copair{\id_{X+E}|\copr_X}\colon (X+E)+E\to X+E$ where 
$\copr_X\colon E\to X+E$ is the coprojection. 
As in Section~\ref{ssec:monad}, 
the category $\cC$ with the monad of exceptions 
provides a model $\cC_\M$ of the decorated logic $\Log_\mon$. 
The name of the decorations can be adapted to the monad of exceptions: 
a constructor is called a \emph{propagator}: 
it may raise an exception but cannot recover from an exception, 
so that it has to propagate all exceptions;
a modifier is called a \emph{catcher}.

For this specific monad, 
it is possible to extend the logic $\Log_\mon$ as $\Log_\exc$,
called the \emph{decorated logic for exceptions}, 
so that $\cC_\M$ can be extended as a model $\cC_\exc$ of $\Log_\exc$. 
First, we get copairs of a propagator and a modifier,  
as in the first part of Fig.~\ref{fig:log-exc}
for the left copairs (the rules for the right copairs are symmetric). 
The interpretation of the left copair $\lcopair{f|g}^\modi:
A_1+A_2\to B$ of $f^\cst:A_1\to B$ and $g^\modi:A_2\to B$
is the copair $\copair{f_1|g_2}:A_1+A_2+E\to B+E$ 
of $f_1:A_1\to B+E$ and $g_2:A_2+E\to B+E$ in~$\cC$.
This is possible because $(A_1+A_2)+E$ 
is canonically isomorphic to $A_1+(A_2+E)$,
whereas for a monad generally $\M(A_1+A_2)$ is not 
isomorphic to $A_1+\M A_2$. 
For instance, the coproduct of $A \cong A+\empt$, 
with coprojections $\id_A^\pure:A\to A$ and $\copa_A^\pure:\empt\to A$, 
gives rise to the left copair $\lcopair{f|g}^\modi:A \to B$ 
of any propagator $f^\cst\colon  A \to B$
with any modifier $g^\modi\colon \empt \to B$,
which is characterized up to strong equations
by $\lcopair{f|g}\eqw f$ and $\lcopair{f|g}\eqs g$.
The construction of $\lcopair{f|g}^\modi$ and its interpretation 
can be illustrated as follows:
$$\xymatrix@R=.5pc@C=7pc{
A \ar[rd]^{f^\cst} \ar[d]_{\id^\pure} & \\
A \ar[r]|{\,\lcopair{f|g}^\modi\,} 
  \ar@{}[ur]|(.3){\eqw} \ar@{}[dr]|(.3){\eqs} & B \\
\empt \ar[u]^{\copa^\pure} \ar[ru]_{g^\modi} & \\
}
\qquad 
\xymatrix@R=.5pc@C=6pc{
A \ar[rd]^{f_1} \ar[d] & \\
A+E \ar[r]|{\,(\lcopair{f|g})_2\,} 
  \ar@{}[ur]|(.3){=} \ar@{}[dr]|(.3){=} & B+E \\
E \ar[u] \ar[ru]_{g_2} & \\
}
$$
Moreover, the rule \rul{effect} expresses the fact that,
when $\M X=X+E$, two modifiers coincide as soon as they coincide 
on ordinary values and on exceptions,
whereas for a monad generally the morphisms $\eta_X\colon X\to \M X$
and $\M \copa_X\colon \M \empt\to \M X$ do not form a coproduct. 
For each set $\Ename$ of \emph{exception names}, 
additional grammar and rules for the logic $\Log_\exc$
are given in Fig.~\ref{fig:log-exc}.
We extend the grammar 
with a type $V_T$, a propagator $\tagg_T^\acc:V_T\to\empt$
and a catcher $\untag_T^\modi:\empt\to V_T$ for each exception name $T$,
and we also extend its rules. 
The logic $\Log_\exc$ obtained performs the \emph{core} operations 
on exceptions: 
the \emph{tagging} operations 
encapsulate an ordinary value into an exception, 
and the \emph{untagging} operations 
recover the ordinary value which has been encapsulated in an exception.
This may be generalized by assuming a hierarchy of exception names  
\cite{DDR13-exc}. 
In Fig.~\ref{fig:log-exc},
the rule $\rul{exc-coprod-u}$ is a decorated rule for coproducts.
It asserts that two functions without argument 
coincide as soon as they coincide on each exception.  
Together with the rule \rul{effect} this implies that two functions 
coincide as soon as they coincide on their argument and on each exception.
\begin{figure}[!ht]   
\renewcommand{\arraystretch}{1.7}
$$ \begin{array}{|l|} 
\hline
\mbox{Additional (left) coproduct rules} \\ 
\quad \rul{l-copair}\quad 
\dfrac{ f_1^\cst\colon  A_1 \to B \quad f_2^\modi\colon  A_2 \to B}
  {\lcopair{f_1|f_2}^\modi\colon  A_1\!+\!A_2 \to B \quad 
  \lcopair{f_1|f_2} \circ \copr_1 \eqw  f_1 \quad 
   \lcopair{f_1|f_2} \circ \copr_2 \eqs f_2 } \\
\quad \rul{l-copair-u}\quad 
\dfrac{g^\modi\scolon  A_1\!\!+\!\!A_2 \sto B \squad 
  f_1^\cst\scolon  A_1 \sto B \squad f_2^\modi\scolon A_2 \sto B \squad 
  g\circ \copr_1 \eqw f_1 \squad g\circ \copr_2 \eqs f_2}
  {g \eqs \lcopair{f_1|f_2}}    \\
\hline
\mbox{Effect rule} \\ 
\quad \rul{effect}\quad 
\dfrac{f,g\colon A \to B \quad f\eqw g \quad 
  f\circ \copa_A \eqs g\circ \copa_A}
  {f\eqs g}  \\ 
\hline
\mbox{Additional grammar (for each $T\in\Ename$)} \\ 
\quad \textrm{Types: } V_T  \\ 
\quad \textrm{Terms: } \tagg_T^\cst\colon V_T\to \empt \mid 
  \untag_T^\modi\colon \empt\to V_T \\ 
\hline
\mbox{Axioms (for each $T\in\Ename$)} \\ 
\quad \untag_T\circ \tagg_T \eqw \id_{V_T} \\
\quad \untag_T\circ\tagg_R\eqw\copa_{V_T}\circ\tagg_R 
  \mbox{ for each }R\ne T, \; R\in\Ename \\ 
\hline
\mbox{A specific coproduct rule} \\ 
\quad \rul{exc-coprod-u} \quad 
\dfrac{f,g\colon \empt\to B \quad \mbox{for all }T\!\in\! \Ename\; 
  f\circ \tagg_T \eqw g\circ \tagg_T}
  {f\eqs g} \\
\hline
\end{array} $$
\renewcommand{\arraystretch}{1}
\caption{From $\Log_\mon$ to $\Log_\exc$: 
additional features for the monad of exceptions} 
\label{fig:log-exc} 
\end{figure}
For each family of objects $(V_T)_{T\in\Ename}$ in $\cC$  
such that $E \cong \sum_{T\in\Ename}V_T $ 
we build a model $\cC_\exc$ of $\Log_\exc$,
which extends the model $\cC_\M$ of $\Log_\mon$
with functions for tagging and untagging the exceptions. 
The types $V_T$ are interpreted as the objects $V_T$ and  
the propagators $\tagg_T^\cst:V_T\to\empt$ as the coprojections 
from $V_T$ to $E$. Then the interpretation of each catcher 
$\untag_T^\modi:\empt\to V_T$ is the function $\untag_T:E\to V_T+E$ 
defined as the cotuple (or case distinction)
of the functions $f_{T,R}:V_R\to V_T+E$  
where $f_{T,T}$ is the coprojection of $V_T$ in $V_T+E$
and $f_{T,R}$ is made of $\tagg_R:V_T\to E$ followed by
the coprojection of $E$ in $V_T+E$ when $R\ne T$.
This can be illustrated, in an informal way, 
as follows: $\tagg_T$ encloses its argument $a$ in a box 
with name $T$, while $\untag_T$ opens every box with name $T$ to recover 
its argument 
and returns every box with name $R \ne T$ without opening it:
 $$ 
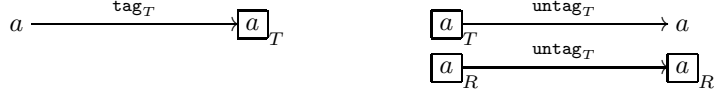
\xymatrix@C=3pc@R=.5pc{
a \ar[rr]^{\tagg_T} && *+[F-]{a} \ar@{}[r]_(.2){T} & \\ 
} \qquad 
\xymatrix@C=3pc@R=.5pc{
*+[F-]{a} \ar@{}[r]_(.2){T} \ar[rr]^{\untag_T} && a & \\ 
*+[F-]{a} \ar@{}[r]_(.2){R} \ar[rr]^{\untag_T} && 
  *+[F-]{a} \ar@{}[r]_(.2){R} & \\
}$$

\subsection{A decorated logic for the comonad of states} 
\label{ssec:state}

Let us assume that there is in $\cC$ 
a distinguished object $S$ called the \emph{object of states}.
The \emph{comonad of states} on $\cC$ is the comonad $(\coM,\eps,\delta)$  
with endofunctor $\coM X=X\times S$, 
its counit $\eps$ is made of the projections 
$\eps_X\colon X\times S\to X$ and its comultiplication $\delta$ 
is defined by 
 $\delta_X=\pair{\id_{X\times S},\pr_X}
\colon X\times S \to (X\times S)\times S $ where 
$\pr_X\colon X\times S \to S$ is the projection. 
This comonad is sometimes called the \emph{product comonad};
it is different from the \emph{costate comonad} or
\emph{store comonad} with endofuntor $ \coM A=S\times A^S$ \cite{GJ12}.
As in Section~\ref{ssec:comonad}, 
the category $\cC$ with the comonad of states 
provides a model $\cC_\coM$ of the decorated logic $\Log_\comon$. 

For this specific comonad, 
it is possible to extend the logic $\Log_\comon$ as $\Log_\sta$,
called the \emph{decorated logic for states}, 
so that $\cC_\coM$ can be extended as a model $\cC_\sta$ of $\Log_\sta$. 
In Fig.~\ref{fig:log-state},
the rule $\rul{st-prod-u}$ is a decorated rule for coproducts.
It asserts that two functions without result 
coincide as soon as they coincide when observed at each location.  
Together with the rule \rul{st-effect} this implies that two functions 
coincide as soon as they return the same value and coincide on each location.
\begin{figure}[!ht]   
\renewcommand{\arraystretch}{1.7}
$$ \begin{array}{|l|} 
\hline
\mbox{Additional (left) product rules} \\ 
\quad \rul{l-pair}\quad 
\dfrac{ f_1^\acc \colon A \to B_1 \quad f_2^\modi \colon  A \to B_2}
  {\lpair{f_1,f_2}^\modi \colon A \to B_1\!\times\! B_2 \quad 
  \pr_1 \circ \lpair{f_1,f_2} \eqw f_1 \quad 
  \pr_2 \circ \lpair{f_1,f_2} \eqs f_2 } \\
\quad \rul{l-pair-u}\quad 
\dfrac{g^\modi \scolon A \to B_1\!\times\! B_2 \squad 
  f_1^\cst \scolon A \sto B_1 \squad f_2^\modi \scolon A \sto B_2 \squad 
  \pr_1 \circ g \eqw f_1 \squad \pr_2 \circ g \eqs f_2}
  {g \eqs \lpair{f_1,f_2}}    \\
\hline
\mbox{Effect rule} \\ 
\quad \rul{st-effect-u}\quad 
\dfrac{f,g\colon A \to B \quad f \eqw g \quad 
  \pa_A \circ f \eqs \pa_A \circ g}
  {f\eqs g}  \\ 
\hline
\mbox{Additional grammar (for each $T\in\Loc$)} \\ 
\quad \textrm{Types: } V_T  \\ 
\quad \textrm{Terms: } \lookup_T^\acc\colon \unit \to V_T \mid 
  \update_T^\modi\colon V_T \to\unit \\ 
\hline
\mbox{Axioms (for each $T\in\Loc$)} \\ 
\quad \lookup_T \circ \update_T \eqw \id_{V_T} \\
\quad \lookup_R \circ \update_T \eqw \lookup_R \circ \pa_{V_T}  
  \mbox{ for each } R\ne T, \; R\in\Loc \\ 
\hline
\mbox{A specific product rule} \\ 
\quad \rul{st-prod-u} \quad 
\dfrac{ f,g\colon A \to \unit \quad \mbox{for all }T\!\in\! \Loc\; 
  \lookup_T \circ f \eqw \lookup_T \circ g}
  {f\eqs g} \\
\hline
\end{array} $$
\renewcommand{\arraystretch}{1}
\caption{From $\Log_\comon$ to $\Log_\sta$: 
additional features for the comonad of states} 
\label{fig:log-state} 
\end{figure}
For each family of objects $(V_T)_{T\in\Loc}$ in $\cC$  
such that $S \cong \prod_{T\in\Loc}V_T $ 
we build a model $\cC_\sta$ of $\Log_\sta$,
which extends the model $\cC_\coM$ of $\Log_\comon$
with functions for looking up and updating the locations. 
The types $V_T$ are interpreted as the objects $V_T$ and  
the accessors $\lookup_T^\acc:\unit\to V_T$ as the projections 
from $S$ to $V_T$. Then the interpretation of each modifier 
$\update_T^\modi:V_T\to\unit$ is the function $\update_T:V_T\times S \to S$ 
defined as the tuple 
of the functions $f_{T,R}:V_T\times S \to V_R$  
where $f_{T,T}$ is the projection of $V_T\times S$ to $V_T$
and $f_{T,R}$ is made of the projection of $V_T\times S$ to $S$
followed by $\lookup_R:S\to V_R$ when $R\ne T$. 

\section{Breaking the duality}
\label{sec:break}

In Section~\ref{ssec:cond-pair} we discuss the behaviour of 
conditionals and $n$-ary operations with respect to effects. 
In Section~\ref{ssec:break-state} 
the decorated logic for states is extended under the assumption 
that $\cC$ is distributive,
and we easily get Theorem~\ref{theorem:state} about 
conditionals and sequential pairs. 
In Section~\ref{ssec:break-exc} 
the decorated logic for exceptions is extended,
in a way which is \emph{not} dual to the extension for states. 
It happens that catchers do not have ``good'' properties with respect to 
conditionals and sequential pairs. 
Thus, we define a new language, 
called the \emph{programmer's language} for exceptions, 
in order to encapsulate the catchers in {\em try-catch} blocks.
This corresponds to the usual way to deal with exceptions in a 
computer language.  
Then, under the assumption that $\cC$ satisfies a limited form of extensivity, 
we get Theorem~\ref{theorem:exc} about 
conditionals and sequential pairs for the programmer's language
for exceptions. 
Note that distributivity and extensivity are related notions \cite{CLW93},
and that each of them breaks the duality between products and coproducts. 

\subsection{Effects: conditionals and sequential pairs} 
\label{ssec:cond-pair}

When there are effects, for a binary operation $f:A_1\times A_2\to B$, 
the fact that the substitution of terms 
$a_1,a_2$ in $f$ is $f\circ \pair{a_1,a_2}$ is no more valid: 
indeed, because of the effects, the result of applying $f$ to 
$a_1,a_2$ may depend on the evaluation order of $a_1$ and $a_2$. 
This means that there is no ``good'' pair $\pair{a_1,a_2}$. 
However, it is usually possible to give a meaning 
to ``$f(a_1,a_2)$ with $a_1$ evaluated before $a_2$'', or 
symmetrically to ``$f(a_1,a_2)$ with $a_2$ evaluated before $a_1$''.
This means that there are ``good'' tuples 
$\pair{a_1\circ v_1,v_2}$ and $\pair{w_1,a_2\circ v_2}$ when 
$v_1$, $v_2$, $w_1$ and $w_2$ are either identities or projections. 
Then, for ``$a_1$ before $a_2$'' one can use 
$\pair{\pr_1,a_2\circ\pr_2}\circ \pair{a_1,\id_A}$
(which coincides with $\pair{a_1,a_2}$ when this pair does exist).
Such a notion of {\em sequential pair} is studied in \cite{DDR11-seqprod}, 
where several effects are considered. 
There are other ways to formalize the fact of 
first evaluating $a_1$ then $a_2$:
for instance by using a strong monad \cite{Moggi91} 
or productors \cite{Tate13};
a comparison with strong monads is done in \cite{DDR11-seqprod}.
$$ \xymatrix@R=.5pc@C=3pc{
&& A_1 && A_1 \\
A \ar[rru]^{a_1} \ar[rrd]_{\id_A} \ar[rr]|{\,\pair{a_1,\id_A}\,} && 
  A_1\times A \ar[u] \ar[d] \ar[rru]^{\pr_1} \ar[rrd]_{a_2\circ\pr_2} 
    \ar[rr]|{\,\pair{\pr_1,a_2\circ\pr_2}\,} 
    \ar@{}[ull]|(.3){=} \ar@{}[dll]|(.3){=} && 
  A_1\times A_2 \ar[u] \ar[d] \ar@{}[ull]|(.3){=} \ar@{}[dll]|(.3){=} 
    \ar[r]^f & B \\
&& A && A_2 \\
}$$ 
For conditionals, the fact that 
{\tt if b then f else g} corresponds to $\copair{f|g}\circ b$
usually remains valid when there are effects.

In this paper, we consider a \emph{language with effects} as
a language with (at least) two levels of terms, 
similar to the \emph{values} and \emph{computations} in \cite{Moggi91}: 
the \emph{pure} terms form the morphisms of a category $\cC$ 
with finite products and coproducts 
and the \emph{general} terms form a larger category $\cC^\gen$  
with the same objects as $\cC$. 
\begin{definition}
\label{defi:cond}
A language with effects {\em is compatible with conditionals} 
when the category $\cC$ has finite coproducts 
and when the copairs of general terms are defined:
for each $f_1^\gen\colon A_1 \to B$ and $f_2^\gen\colon A_2 \to B$ 
there exists a unique $\copair{f_1|f_2}^\gen\colon A_1+A_2 \to B$
such that $\copair{f_1|f_2} \circ \copr_1 =  f_1 $
and $\copair{f_1|f_2} \circ \copr_2 =  f_2$
(where $\copr_1^\pure\colon A_1 \to A_1+A_2 $ and 
$\copr_2^\pure\colon A_2 \to A_1+A_2$ are the coprojections).
\end{definition}
\begin{definition}
\label{defi:seqprod}
Let $\geqaux$ be a relation between pure terms and general terms 
which is the equality when both terms are pure. 
A language with effects {\em is compatible with sequential pairs,
with respect to $\geqaux$},  
when the category $\cC$ has finite products 
and when the left and right pairs of a pure term and a general term are defined,
in the following sense: 
for each $f_1^\pure\colon A \to B_1$ and $f_2^\gen\colon A \to B_2$ 
there exists a unique $\lpair{f_1,f_2}^\gen\colon A \to B_1\times B_2$
such that $\pr_1\circ\lpair{f_1,f_2} \geqaux f_1 $ and  
$\pr_2\circ\lpair{f_1,f_2} = f_2$ 
(where $\pr_1^\pure\colon  B_1\times B_2\to B_1 $ and 
$\pr_2^\pure\colon  B_1\times B_2\to B_2 $ are the projections),
and symmetrically for $\rpair{f_1^\gen,f_2^\pure}^\gen$.
\end{definition}

\subsection{States} 
\label{ssec:break-state}

Let us assume that the category $\cC$ is distributive. 
This means that the canonical morphism from 
$A\times B+A\times C$ to $A\times (B+C)$ is an isomorphism. 
Then we get new decorations for the coproduct rules, because 
the copair of two modifiers now exists, see Fig.~\ref{fig:log-state-plus}. 
The interpretation of the modifier $\copair{f|g}^\modi$,
when both $f^\modi$ and $g^\modi$ are modifiers, 
is the composition of the inverse of 
the canonical morphism $(A_1\times S)+(A_2\times S)\to (A_1+A_2)\times S$
with $\copair{f_2|g_2}\colon (A_1\times S)+(A_2\times S)\to B\times S$.
\begin{figure}[!ht]   
\renewcommand{\arraystretch}{2}
$$ \begin{array}{|l|} 
\hline
\mbox{Additional coproduct rules} \\
\quad \rul{copair}  \quad 
\dfrac{ f_1^\modi\colon  A_1 \to B \quad f_2^\modi\colon  A_2 \to B}
  {\copair{f_1|f_2}^\modi\colon  A_1\!+\!A_2 \to B  \quad 
   \copair{f_1|f_2} \circ \copr_1 \eqs  f_1 \quad 
   \copair{f_1|f_2} \circ \copr_2 \eqs  f_2  } \quad \\
\quad \rul{copair-u}  \; 
\dfrac{f_1^\modi\scolon  A_1 \sto B \squad 
  f_2^\modi\scolon  A_2 \sto B \squad 
  g^\modi\scolon  A_1\!\splus\!A_2 \sto B \squad 
  g\circ \copr_1 \eqs f_1 \squad 
  g\circ \copr_2 \eqs f_2 }
  {g \eqs \copair{f_1|f_2} } \\
\hline
\end{array} $$
\renewcommand{\arraystretch}{1}
\caption{From $\Log_\sta$ to $\Log_\sta^+$: additional rules for states, 
when $\cC$ is distributive} 
\label{fig:log-state-plus} 
\end{figure}
\begin{theorem}
\label{theorem:state}
Let us consider the language for states  
with modifiers as general terms (decoration $g=2$).   
When the category $\cC$ is distributive, 
the language for states  
is compatible with conditionals and sequential pairs 
with respect to $\eqw$.
\end{theorem}
\begin{proof}
The left and right pairs of an accessor and a modifier 
in the logic $\Log_\sta$ (Fig.~\ref{fig:log-state}) 
provide sequential pairs. 
The rules for copairs 
in the logic $\Log_\sta^+$ (Fig.~\ref{fig:log-state-plus}) 
provide conditionals.
\end{proof}
\begin{remark}
\label{remark:state}
An advantage of 
using the comonad of states $X\times S$ rather than the 
usual monad of states $(X\times S)^S$ is that 
sequential pairs for states are defined without any new ingredient: 
no kind of strength, in contrast with the approach using 
the strong monad of states $(A\times S)^S$ \cite{Moggi91},
and no ``external'' decoration for equations, in contrast with 
\cite{DDR11-seqprod}. 
\end{remark}

\subsection{Exceptions} 
\label{ssec:break-exc}

Since we do not assume that the category $\cC$ is codistributive 
we do not get pairs of catchers 
in a way dual to the copairs of modifiers for states.
In fact the decorated logic $\Log_\exc$ for exceptions, 
with the core operations for tagging and untagging, 
remains \emph{private}, while 
there is a \emph{programmer's} language, 
which is \emph{public}, with no direct access to the catchers.
The programmer's language for exceptions
provides the operations for \emph{raising} and \emph{handling} exceptions,
which are defined in terms of the core operations. 
This language does not include the private tagging and untagging operations,
but the public $\throw$ and $\try/\catch$ constructions,
which are defined in terms of $\tagg$ and $\untag$. 
It has no catcher: 
the only way to catch an exception is by using a $\try/\catch$ expression, 
which itself propagates exceptions. 
This corresponds to the usual mechanism of exceptions 
in programming languages. 
For the sake of simplicity we assume that only one type of exception is
handled in a $\try/\catch$ expression,
the general case is treated in Appendix~\ref{app:catch}. 

The main ingredients for building the programmer's language 
from the core language are the coproducts $A\cong A+\empt$ 
and the fact of decorating the composition:
in addition to the basic composition ``$\circ$'' 
we introduce a second composition, 
called the {\it propagator composition} and denoted ``$\odot$'', 
subject to the rules in Fig.~\ref{fig:log-exc-plus}. 
Both compositions ``$\circ$'' and ``$\odot$'' coincide on propagators, 
but they are interpreted differently when a propagator is composed 
with a modifier. 
This is an instance of the two ways to compose oblique morphisms 
related to an adjunction \cite{Munch14}. 
\begin{remark}
\label{remark:composition}
In fact, this new composition can be defined for any monad,
but until now it has not been needed: 
Let $f^\cst:A\to B$ and $k^\modi:B\to C$ then 
$(k \odot f)^\cst: A \to C$ is interpreted as $k_2\circ f_1:A \to \M C$;  
then it can be checked that 
$f \eqw g$  if and only if $ \id \odot f\eqs \id \odot g$. 
In contrast, $(k \circ f)^\modi: A \to C$ is interpreted as 
$k_2\circ f_2 = k_2\circ \mu_B\circ \M f_1:\M A \to \M C$. 
Dually, such a new composition could be defined for any comonad. 
\end{remark}
Now, we come back to exceptions and we define the 
$\throw$ and $\try/\catch$ constructions. 
\begin{definition}
\label{defi:exc}
For each type $B$ and each exception name $T$, the propagator
$ \throw_{B,T}^\cst$ is: 
 $$ \throw_{B,T}^\cst = \copa_B^\pure \circ \tagg_T^\cst \colon V_T\to B $$
For each each propagator $f^\cst\colon A\to B$, each exception name $T$ and 
each propagator $g^\cst\colon V_T\to B$, 
the propagator $\try(f)\catch(T\To g)^\cst$ is defined 
as follows, in two steps:
\renewcommand{\arraystretch}{1.3}
$$ \begin{array}{l}
\catch(T\To g)^\modi = 
  \copair{\; g^\cst \;|\; \copa_B^\pure \;}^\cst \circ\untag_T^\modi 
  \colon  \empt\to B \\ 
\try(f)\catch(T\To g)^\cst = 
  \lcopair{\; \id_B \;|\; \catch(T\To g) \;}^\modi \odot f^\cst 
  \colon A\to B \\
\end{array}$$
\renewcommand{\arraystretch}{1}
\end{definition}
This means that raising an exception with name $T$ in a type $B$ 
consists in tagging the given ordinary value (in $V_T$) 
as an exception and coerce it to~$B$. 
For handling an exception, 
the intermediate expression $\catch(T\To g)$ is a private catcher 
while the expression $\try(f)\catch(T\To g)$ is a public propagator: 
the propagator composition ``$\odot$'' 
prevents this expression from catching exceptions 
with name $T$ which might have been raised before 
the $\try(f)\catch(T\To g)$ block is considered. 
The definition of $\try(f)\catch(T\To g)$ corresponds to 
the Java mechanism for exceptions \cite{java,Jacobs01}, 
which may be described by the control flow in Fig.~\ref{fig:flow},
where ``$\isexc$'' means ``\emph{is this value an exception?}'',
an \emph{abrupt} termination returns an uncaught exception 
and a \emph{normal} termination returns an ordinary value.
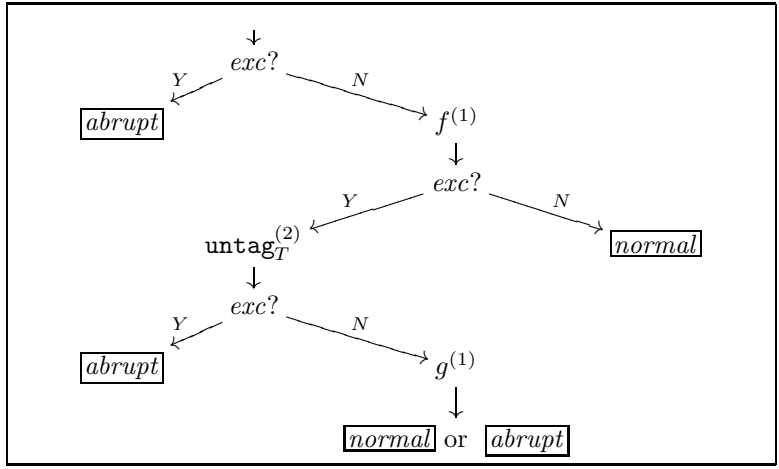
\begin{figure}[!ht]   
$$ \begin{array}{|c|}
\hline
\qquad \xymatrix@C=.8pc@R=.6pc{
& \ar[d] && \\
& \isexc \ar[ld]_{Y}\ar[rd]^{N} && \\
\abr && f^\cst \ar[d] & \\
&& \isexc \ar[ld]_{Y}\ar[rd]^{N} & \\
& \untag_T^\modi \ar[d] && \nor \\
& \isexc \ar[ld]_{Y}\ar[rd]^{N} && \\ 
\abr && g^\cst \ar[d] & \\
&& \txt{\nor \mbox{ or } \abr} \\
} \qquad \\ 
\hline
\end{array} $$
\caption{The control flow for $\try(f)\catch(T\To g)$} 
\label{fig:flow} 
\end{figure}
Now, let us assume that the category $\cC$ is {\em extensive
with respect to $E$}, by which we mean that  
the pullbacks of the coprojections $\copr_1\colon B\to B+E$ 
and $\copr_2\colon E\to B+E$ 
along an arbitrary morphism $f\colon A\to B+E$ exist 
and form a coproduct $A=\cD_f+\cE_f$: 
$$ \xymatrix@R=1pc@C=3pc{
\cD_f \ar[d]_{i_f} \ar[r]^{f_\normal} & B \ar[d]^{\copr_1} \\
A \ar[r]^{f_1} & B+E \\
\cE_f \ar[u]^{j_f} \ar[r]^{f_\abrupt} & E \ar[u]_{\copr_2} \\
}$$
Informally, this implies that any morphism $f_1\colon A\to B+E$ can be seen 
as a partial morphism form $A$ to $B$ with domain of definition 
the vertex $\cD_f$ of the pullback on $\copr_1$ and $f_1$.  
We get a decorated logic $\Log_\exc^+$ by extending $\Log_\exc$ 
with the propagator composition 
and with left pairs (and right ones, omitted here)
as in Fig.~\ref{fig:log-exc-plus}. 
\begin{figure}[!ht]   
\renewcommand{\arraystretch}{1.7}
$$ \begin{array}{|l|} 
\hline
\mbox{Propagator composition} \\ 
\quad \rul{prop-comp} \quad 
  \dfrac{f^\cst\colon A\to B \quad g^\modi\colon B\to C}
    {(g\odot f)^\cst \colon A\to C \quad 
    g\odot f \eqw g\circ f} \quad \\ 
\hline
\mbox{Additional (left) product rules} \\ 
\quad \rul{l-pair}  \quad 
  \dfrac{f_1^\pure\colon A \to B_1 \quad f_2^\cst\colon A \to B_2}
    {\lpair{f_1,f_2}^\cst\colon  A\to B_1\stimes B_2 \quad
    \pr_1\circ\lpair{f_1,f_2} \leqaux f_1 \quad 
    \pr_2\circ\lpair{f_1,f_2} \eqw f_2 }  \\
\quad \rul{l-pair-u} \quad 
  \dfrac{f_1^\pure\scolon A\!\sto\! B_1 \squad
    f_2^\cst\scolon A\!\sto\! B_2 \squad 
    g^\cst\scolon A\!\sto\! B_1\stimes B_2 \squad 
    \pr_1\circ g\leqaux f_1 \squad 
    \pr_2\circ g\eqs f_2 }
    {g \eqs \lpair{f_1,f_2}} \\
\hline 
\end{array}$$
\renewcommand{\arraystretch}{1}
\caption{From $\Log_\exc$ to $\Log_\exc^+$: additional rules for exceptions, 
when $\cC$ is extensive wrt $E$}
\label{fig:log-exc-plus} 
\end{figure}
We define a relation $\geqaux$ between pure terms and propagators,  
which can be seen as (a restriction of) the usual order 
between partial functions. 
\begin{definition}
\label{defi:exc-order}
Let $v^\pure\colon A\to B$ be a pure term 
and $f^\cst\colon A\to B$ a propagator, 
corresponding respectively to 
$v_0\colon A \to B$ and $f_1\colon A\to B+E$ in $\cC$.  
Then $v^\pure\geqaux f^\cst$ if and only if
the restrictions of $v_0$ and $f_1$ to the domain of definition of $f_1$ 
coincide, which means, if and only if $v_0\circ i_f=f_{\normal}\colon \cD_f\to B$.
\end{definition}
Now we can interpret the left pair of a pure term and a propagator.
\begin{definition}
\label{defi:exc-lpair}
Let $v^\pure\colon A\to B_1$ be a pure term 
and $f^\cst\colon A\to B_2$ a propagator, 
corresponding respectively to 
$v_0\colon A \to B_1$ and $f_1\colon A\to B_2+E$ in $\cC$.  
Let $h_\normal=\pair{v\circ i_f, f_\normal} \colon \cD_g \to B_1\times B_2$
and $h_\abrupt=f_\abrupt\colon \cE_g \to E$,
then let $h= h_\normal+h_\abrupt \colon A \to (B_1\times B_2)+E$ in $\cC$. 
The morphism $h$ in $\cC$ corresponds to a propagator 
$h^\cst\colon A\to B_1\times B_2$,
which is the interpretation of the \emph{left pair} $\lpair{v,f}^\cst$
of $v^\pure$ and $f^\cst$. 
\end{definition}
It is easy to check that indeed $h^\cst$ 
satisfies the properties required of left pairs in Fig.~\ref{fig:log-exc-plus}. 
The right pair of a propagator and a pure term is defined in a symmetric way.
It can easily be checked that the core language for exceptions 
with catchers as general terms (decoration $g=2$) is not compatible 
with conditionals and sequential pairs 
(with respect to any relation $\geqaux$).
\begin{theorem}
\label{theorem:exc}
Let us consider the programmer's language for exceptions 
with propagators as general terms (decoration $g=1$).   
When the category $\cC$ is extensive with respect to $E$, 
the programmer's language for exceptions 
is compatible with conditionals and sequential pairs 
with respect to $\geqaux$ as in Definition~\ref{defi:exc-order}.
\end{theorem}
\begin{proof}
The left and right pairs of a pure term and a propagator 
in the logic $\Log_\exc^+$ (Fig.~\ref{fig:log-exc-plus}) 
provide sequential pairs. 
The rules for copairs 
in the logic $\Log_\mon$ (Fig.~\ref{fig:log-mon}) 
provide conditionals.
\end{proof}




\newpage
\appendix 
\section{The decorated logic for a monad} 
\label{app:logic}

\small
The decorated logic $\Log_\mon$ for a monad when $\cC$ is bicartesian. 
$$ \begin{array}{|l|} 
\hline
\mbox{Grammar} \\ 
\quad \textrm{Types: } t::= 
 A\mid B\mid \dots\mid 
 t\times t\mid \unit\mid
 t+t\mid\empt  \\
\quad \textrm{Terms: } f::=   \id_t\mid f\circ f\mid  
 \pair{f,f} \mid \pr_{t,t,1}\mid \pr_{t,t,2}\mid \pa_t \mid
 \copair{f|f}\mid \copr_{t,t,1}\mid\copr_{t,t,2}\mid\copa_t \quad \\
\quad \textrm{Decoration for terms: } \dec::= \pure \mid \cst \mid \modi \\ 
\quad \textrm{Equations: } e::=  f\eqs f  \mid f\eqw f \\  
\hline
\mbox{Conversion rules} \\ 
\quad \dfrac{f^\pure}{f^\acc} \qquad \dfrac{f^\acc}{f^\modi} 
  \qquad \dfrac{f^\dec\eqs g^\decp}{f\eqw g} \qquad
  \dfrac{f^\dec\eqw g^\decp}{f\eqs g} \mbox{ if } d,d'\leq 1 \\
\hline
\mbox{Equivalence rules} \\ 
\quad \rul{s-refl} \quad 
  \dfrac{f^\dec}{f \eqs f} \qquad
\rul{s-sym} \quad 
  \dfrac{f^\dec \eqs g^\decp}{g \eqs f}  \qquad
\rul{s-trans} \quad 
  \dfrac{f^\dec \eqs g^\decp \squad g^\decp \eqs h^\decpp}{f \eqs h}  \\ 
\quad \rul{w-refl} \quad 
  \dfrac{f^\dec}{f \eqw f} \qquad
\rul{w-sym} \quad 
  \dfrac{f^\dec \eqw g^\decp}{g^\decp \eqw f}  \qquad
\rul{w-trans} \quad 
  \dfrac{f^\dec \eqw g^\decp \squad g^\decp \eqw h^\decpp}{f \eqw h}  \\ 
\hline
\mbox{Categorical rules} \\ 
\quad \rul{id} \quad 
  \dfrac{A}{\id_A^\pure\colon A\to A } \qquad 
\rul{comp} \quad 
  \dfrac{f^\dec\colon A\to B \quad g^\decp\colon B\to C}
    {(g\circ f)^{(max(d,d'))} \colon A\to C}  \\
\quad \rul{id-source} \quad 
  \dfrac{f^\dec\colon A\to B}{f\circ \id_A \eqs f} \qquad 
\rul{id-target} \quad 
  \dfrac{f^\dec\colon A\to B}{\id_B\circ f \eqs f} \\
\quad \rul{assoc} \quad 
  \dfrac{f^\dec\colon A\to B \squad g^\decp\colon B\to C \squad 
     h^\decpp\colon C\to D}
  {h\circ (g\circ f) \eqs (h\circ g)\circ f}  \\
\hline
\mbox{Congruence rules} \\ 
\quad \rul{s-repl} \;
  \dfrac{f_1^\done\eqs f_2^\dtwo\colon A\to B \squad g^\dec\colon B\to C}
    {g\circ f_1 \eqs g\circ f_2 }  \quad
\rul{s-subs} \;
  \dfrac{f^\dec\colon A\to B \squad g_1^\done\eqs g_2^\dtwo\colon B\to C}
    {g_1 \circ f \eqs g_2\circ f } \\
\quad \rul{w-repl}  \;
  \dfrac{f_1^\done\eqw f_2^\dtwo\colon A\to B \squad g^\dec\colon B\to C}
    {g\circ f_1 \eqw g\circ f_2 }  \quad
\rul{w-subs} \;
  \dfrac{f^\pure\colon A\to B \squad g_1^\done \eqw g_2^\dtwo\colon B\to C}
    {g_1 \circ f \eqw g_2\circ f } \\
\hline
\mbox{Product rules} \\ 
\quad \rul{prod} \quad 
  \dfrac{B_1 \quad B_2 }
    {\pr_1^\pure\colon B_1\stimes B_2 \to B_1 \quad 
    \pr_2^\pure\colon B_1\stimes B_2 \to B_2}  \\ 
\quad \rul{pair}  \quad 
  \dfrac{ f_1^\pure\colon  A \to B_1 \quad f_2^\pure\colon  A \to B_2}
    {\pair{f_1,f_2}^\pure\colon  A\to B_1\stimes B_2 \quad
    \pr_1\circ\pair{f_1,f_2} \eqs f_1 \quad 
    \pr_2\circ\pair{f_1,f_2} \eqs f_2 }  \\
\quad \rul{pair-u} \quad 
  \dfrac{f_1^\pure\scolon A\!\sto\! B_1 \squad
    f_2^\pure\scolon A\!\sto\! B_2 \squad 
    g^\pure\scolon A\!\sto\! B_1\stimes B_2 
    \squad \pr_1\circ g\eqs f_1 \squad \pr_2\circ g\eqs f_2 }
    {g \eqs \pair{f_1,f_2}} \\
\quad \rul{final} \quad 
  \dfrac{A}{\pa_A^\pure\colon A\to \unit} \qquad
\rul{final-u} \quad 
  \dfrac{f^\pure \colon A\to \unit}{f \eqs \pa_A} \\  
\hline
\mbox{Coproduct rules} \\ 
\quad \rul{coprod} \quad 
  \dfrac{A_1 \quad A_2}
    {\copr_1^\pure\colon A_1\to A_1\splus A_2 \quad 
    \copr_2^\pure\colon A_2\to A_1\splus A_2 } \\ 
\quad \rul{copair} \; 
  \dfrac{ f_1^\done\colon  A_1 \to B \quad 
    f_2^\dtwo\colon  A_2 \to B}
    {\copair{f_1|f_2}^{(max(d_1,d_2))} \colon A_1\splus A_2 \to B \quad
    \copair{f_1|f_2} \circ \copr_1 \eqs  f_1 \quad 
    \copair{f_1|f_2} \circ \copr_2 \eqs  f_2  } (d_1,d_2\leq1) \\
\quad \rul{copair-u} \; 
  \dfrac{f_1^\done\scolon A_1 \sto B \squad 
    f_2^\dtwo\scolon A_2 \sto B \squad 
    g^\dec\scolon A_1\!\splus\! A_2 \sto B \squad 
    g\scirc \copr_1 \!\eqs\! f_1 \squad 
    g\scirc \copr_2 \!\eqs\! f_2 }
    {g \!\eqs\! \copair{f_1|f_2}} (d_1,d_2,d\!\leq\!1) \\  
\quad \rul{initial} \quad 
  \dfrac{B}{\copa_B^\pure \colon \empt\to B} \qquad
\rul{initial-u} \quad 
  \dfrac{f^\dec\colon \empt\to B}{f \eqw \copa_B} \\ 
\hline 
\end{array}$$

\normalsize

\section{Catching several exception names} 
\label{app:catch}

The handling process is easily extended to several exception names, as follows. 
The index $T_i$ is simplified as $i$: $V_i=V_{T_i}$, 
$\tagg_i=\tagg_{T_i}$, $\untag_i=\untag_{T_i}$. 
\begin{definition}
\label{defi:exc-multi}
For each each propagator $f^\cst\colon A\to B$, 
each list of exception names $(T_1,\dots,T_n)$ and 
each propagators $g_j^\cst\colon V_i\to B$ for $i=1,\dots,n$, 
the propagator 
$ \try(f)\catchn{T_1}{g_1}{T_n}{g_n}^\cst \colon A\to B$ 
is defined as follows, in two steps:
\begin{itemize}
\item the catcher $\catchn{T_1}{g_1}{T_n}{g_n}^\modi : \empt\to B$ 
is obtained by setting $i=1$ in the family of catchers 
$k_i^\modi = \catchn{T_i}{g_i}{T_n}{g_n} : \empt\to B$ 
(for $i=1,\dots,n$) which are defined recursively by: 
  $$ k_i^\modi \;=\;  
    \begin{cases} 
      \copair{\; g_n^\cst \;|\; \copa_B^\pure \;}^\cst \circ\untag_n^\modi &
             \mbox{ when } i=n \\
      \lcopair{\; g_i^\cst \;|\; k_{i+1}^\modi \;}^\modi \circ \untag_i^\modi & 
		   \mbox{ when } i< n \\
   \end{cases} $$
\item then the required propagator is: 
\renewcommand{\arraystretch}{0.5}
$$ \begin{array}{l}
\try(f)\catchn{T_1}{g_1}{T_n}{g_n}^\cst = 
\\ \qquad\qquad  \lcopair{\; \id_B \;|\; \catchn{T_1}{g_1}{T_n}{g_n} \;}^\modi 
\odot f^\cst   \colon A\to B \\
\end{array}$$
\renewcommand{\arraystretch}{0.5}
\end{itemize}
\end{definition}

The handling process is also easily extended to all exception names.
This \emph{catch-all} construction
is similar to \texttt{catch(...)} in C++ or to 
\texttt{(except, else)} in Python.
We add a catcher $\untag_{\all}^{\modi}\colon \empt\to \unit$
with the equations 
$$\untag_{\all} \circ \tagg_{T} \eqw \pa_T $$
for every exception name $T$,
which means that $\untag_{\all}$ catches exceptions of the form $\tagg_{T}(a)$
for every $T$ and forgets the value $a$.
\begin{definition}
\label{defi:exc-all} 
For each propagators $f^\cst\colon A\to B$ and $g^\cst\colon \unit\to B$,
the propagator $ \try(f)\catch(\all\To g)^\cst\colon A\to B $ is:
$$ \try(f)\catch(\all\To g)^\cst = 
  \lcopair{\; \id_B \,|\, g\circ \untag_{\all} \;}^\modi \odot f^\cst
\colon A\to B $$
\end{definition}
The interpretation of $ \try(f)\catch(\all\To g)$ is 
``\emph{handle the exception $e$ raised in $f$, if any, with $g$}''. 
This may be combined with other catchers, and 
every catcher following a {\it catch-all} is
syntactically allowed, but never executed.

\end{document}